\documentclass[12pt]{article}

\usepackage{a4wide}
\usepackage{amsmath,amssymb,amsfonts,amsthm} 
\usepackage{color}

\usepackage{mathrsfs,latexsym  }

\newtheorem{Thm}{Theorem}[section]

\newtheorem{Lem}[Thm]{Lemma}
\newtheorem{Prop}[Thm]{Proposition}

\newcommand{\II}{{\mathbb I}}
\newcommand{\CC}{{\mathbb C}}

\newcommand{\RR}{{\mathbb R}}

\newcommand{\NN}{{\mathbb N}}

\newcommand{\ProrL}{{\mathscr{L}_+^\uparrow}}

\newcommand{\ProrW}{{\mathscr{W}_+^\uparrow}}

\numberwithin{equation}{section}

\newcommand{\Ac}{{\mathcal A}}
\newcommand{\Bc}{{\mathcal B}}

\newcommand{\Dc}{{\mathcal D}}

\newcommand{\Hc}{{\mathcal H}}

\newcommand{\Ba}{{\cal B}}

\newcommand{\Oa}{{\cal O}}

\newcommand{\Wa}{{\cal W}}

\newcommand{\yb}{{\boldsymbol{y}}}

\newcommand{\Mc}{{\mathcal{M}}}

\newcommand{\Wc}{{\mathcal{W}}}

\def\bio{{\mbox{\boldmath $o$}}}
\def\bis{{\mbox{\boldmath $s$}}}
\def\biss{{\mbox{\footnotesize \boldmath $s$}}}
\def\binull{{\mbox{\boldmath $0$}}}
\def\bisnull{{\mbox{\footnotesize \boldmath $\, o$}}}
\def\bix{{\mbox{\boldmath $x$}}}
\def\bisx{{\mbox{\footnotesize \boldmath $x$}}}
\def\binull{{\mbox{\boldmath $0$}}}
\def\biy{{\mbox{\boldmath $y$}}}
\def\bisy{{\mbox{\footnotesize \boldmath $\, y$}}}

\def\biL{{\mbox{\boldmath $L$}}}

\def\bieta{{\mbox{\boldmath $\eta$}}}

\def\Dbox{{\Dc {}_{\scriptscriptstyle \square}}}

\def\sfa{{\sf a}}

\def\ie{{\it i.e.\ }}
\def\viz{{\it viz.\ }}

\begin{document}

%

\title{Macroscopic aspects of the Unruh effect}
\author{Detlev Buchholz${}^{(a)}$ \ \ and \ \ Rainer Verch${}^{(b)}$ \\[20pt]
\small 
${}^{(a)}$ Institut f\"ur Theoretische Physik, Universit\"at G\"ottingen, \\
\small Friedrich-Hund-Platz 1, 37077 G\"ottingen, Germany\\[5pt]
\small
${}^{(b)}$ Institut f\"ur Theoretische Physik, Universit\"at Leipzig, \\
\small Br\"uderstr.\ 16, 04103 Leipzig, Germany}
\date{}

\maketitle

\begin{abstract}
\noindent Macroscopic concepts pertaining 
to the Unruh effect are elaborated and used 
to clarify its physical manifestations. 
Based on a description of the motion of accelerated,    
spatially extended laboratories in Minkowski space 
in terms of Poincar\'e 
transformations, it is shown that, from a macroscopic 
perspective, an accelerated observer will not register 
with his measuring instruments any global thermal effects of 
acceleration in the inertial (Minkowskian) 
vacuum state. 
As is explained, this result is not 
in conflict with the well--known fact that microscopic 
probes used as thermometers  
respond non--trivially to acceleration if coupled 
to the vacuum. But this response cannot be 
interpreted as the effect of some exchange of thermal energy 
with a gas surrounding the observer; in fact, 
it is induced by the measuring process itself.
It is also shown that genuine equilibrium states 
in a uniformly accelerated laboratory cannot be spatially
homogeneous. In particular, these states 
coincide with the homogeneous inertial vacuum at sufficiently large 
distances from the horizon of the observer and consequently have 
the same (zero) temperature there. The analysis is carried out 
in the theory of a free massless scalar field;
however the conclusion that the Unruh effect is not 
of a macroscopic thermal origin is generally valid.
\end{abstract}

\section{Introduction}
\setcounter{equation}{0}

In spite of the fact that the computational aspects of the 
Unruh effect have been extensively studied 
and are by now well understood \cite{Unruh,Ta,Wald}, 
there has not yet emerged a consensus on its 
proper physical interpretation, cf.\ for example \cite{CoFo,Earman,BuSo}.
It seems desirable to settle this matter in order to 
gain clarity as to what the theory actually predicts about  
the as yet unattainable experimental situation~\cite{CrHiMa}. 
In the present article
we elaborate theoretical concepts whose operational 
significance in the context of the Unruh effect is evident 
but whose mathematical description requires some care. 
Having clarified these concepts we will reconsider the 
theoretical predictions pertaining to the Unruh effect and  
come to conclusions that corroborate its dissenting 
interpretation brought forward in \cite{BuSo}.   

\medskip 
\noindent (i) The first point in the context of the Unruh effect
which requires some thoughts is the
fact that measuring devices are spatially extended. 
The popular idealisation
that measurements are performed along world lines with 
comoving clocks attached to each line, indicating   
the respective proper time, corresponds neither to a
realistic experimental situation nor to a meaningful 
theoretical hypothesis since it leads to difficulties 
when dealing with extended observables in the 
Heisenberg picture. Conceptual problems caused by this   
overidealized treatment were already observed by
Bell, Hughes and Leinaas in \cite{BeHuLe}. 
We will therefore rely on 
the following more realistic scenario: A Minkowski space based 
observer enters with his clock a laboratory in a 
spacecraft that, at a given time, is at rest and then undergoes acceleration. 
The experimental 
equipment that he takes along is mounted to the walls of the laboratory
at rest relative to him 
at all times, measured with his clock. Thus there is only one relevant 
time scale within the spatially extended 
laboratory. It is fixed by the clock of the 
observer and the interpretation of observations 
made with his measuring instruments rests on that time scale.
It is understood that the walls of the laboratory are to be rigid in order 
to compensate tidal forces and allow for this standard    
experimental situation. We will discuss this point in the 
subsequent section, where we show that under mild constraints 
on the forces any motion of the laboratory can be described by a 
family of Poincar\'e transformations that is parametrised by the
eigentime of the observer. 

\medskip
\noindent (ii) Next, there is the question at which 
scales observations are performed. 
Being interested in the Unruh effect, the
observer will be led to analyse the macroscopic thermal properties of 
the vacuum in his laboratory.
He does this by subsequent measurements of his observables
in order to suppress microscopic fluctuations,
thereby enhancing those 
features of the state that prevail at asymptotic times and 
hence can be interpreted as macroscopic (superselected) properties
of the state 
\cite{He,Se2}. We will show in Sec.~3 that, irrespective of the 
motion performed by the laboratory, all observables 
form central sequences at asymptotic times  in the 
state space of the Minkowski
vacuum and have sharp (non--fluctuating) limits. 
The numerical values of these limits do not depend on the 
details of the motion and coincide with those found by an inertial 
observer. In other words, an accelerated observer will not 
register any macroscopic thermal effects caused by acceleration.

\medskip
\noindent (iii) Having established the absence of thermal 
effects in the inertial vacuum which are caused by acceleration, 
the observer 
can prepare and study other states that are in equilibrium in his 
laboratory and analyse their macroscopic thermal   
properties. Restricting attention to the case of constant
\mbox{acceleration} we will show in Sec.~4 that all genuine equilibrium
states, characterised by the \mbox{KMS-condition} with respect to the
accelerating dynamics, exhibit macroscopic 
properties that depend on the distance from the 
(apparent) horizon of the observer. Moreover, all equilibrium states 
coincide with the homogeneous 
inertial vacuum if restricted to compact regions 
of arbitrary size at sufficiently
large distances from this horizon. Thus all equilibrium states 
have the same (zero) temperature in these remote regions. 
These facts, akin to the classical Tolman--Ehrenfest effect 
\cite{To,EhTo}, corroborate the assertion in \cite{BuSo} that the 
equilibrium parameters $T$ characterising KMS states 
cannot be interpreted as ``local temperature''  
in the presence of acceleration;    
instead, they subsume information about the 
relation between the temperature and the acceleration that must 
prevail in order to accomplish global equilibrium. 

\medskip
\noindent (iv) Finally, the uniformly accelerated observer can place  
microscopic probes into equilibrium states, prepared in his laboratory, and 
determine at sufficiently large times the accumulated impact  
on the probes caused by the interaction with the  
states. Within the theoretical framework this situation 
is described by Pauli--Fierz type models of a finite dimensional 
quantum system that is weakly coupled to a macroscopic KMS state. 
In these models one can show under mild conditions 
on the underlying dynamics and couplings, applying  
to the case at hand, that the composed system approaches 
at large times a KMS state with the same equilibrium 
parameter $T$ as that of the initially unperturbed KMS state, 
cf.~\cite{DeJa} and references quoted there. A detailed
discussion of the Unruh effect in this setting has been 
given in~\cite{BiMe}. From a physical 
point of view this ``return to equilibrium'' is not  
surprising. Yet this feature does not imply that  
the equilibrium parameter~$T$ displayed by the microscopic 
probe can be interpreted as temperature of the macroscopic state, 
even if one corrects this reading by redshift factors depending
on the position of the probe within the laboratory.
For the interaction between the probe and the state 
does not only induce the desired exchange of thermal energy between the 
two systems; it also creates from this state 
additional excitations due to its local nature 
(cf.\ Reeh--Schlieder theorem \cite[Sec.\ II.5.3]{Haag}).  
These excitations transmit additional energy to the probe,
proportional to the acceleration, leading to higher 
values of $T$ than the actual temperature of the macroscopic 
equilibrium state. So 
probes may not be regarded as perfect thermometers in the presence of 
acceleration since inevitable quantum effects in the measuring process 
affect their readings.

\medskip
We conclude this introduction by defining our notation 
and presenting a model that is commonly used in discussions 
of the Unruh effect. Throughout this article we use
units where $c=\hbar=k=1$. We consider 
four dimensional Minkowski space $\Mc = (\RR^4 , g)$
with proper coordinates $x = (x_0, \bix)$ and metric 
$g$ fixed by the Lorentz scalar product 
$x  \cdot  y \doteq x_0 y_0 - \bix \biy$,
where $\bix \biy$ denotes the Euclidean scalar product of the 
spatial components $\bix, \biy$. 
On~$\Mc$ there acts the extended Poincar\'e group (Weyl group) 
$\ProrW \doteq \RR^4 \rtimes (\ProrL \times \RR_+)$ consisting of 
spacetime translations, proper orthochronous Lorentz transformations
and dilations.  
The product of its elements $\Upsilon = (y,\Lambda,\lambda)$ 
is defined by 
\mbox{$ \Upsilon_1  \Upsilon_2 =
(y_1, \Lambda_1, \lambda_1) (y_2, \Lambda_2, \lambda_2)
\doteq 
(y_1 + \Lambda_1 \lambda_1 y_2,  \Lambda_1 \Lambda_2, \lambda_1  \lambda_2)$}
in an obvious notation. 

\bigskip
The model under consideration is the theory of a free massless 
scalar and hermitian field $\phi$. Since we consider 
several inequivalent
representations of this field that are induced by 
different equilibrium 
states it is convenient to present the theory in terms of 
bounded functions 
\mbox{$W(f) = \exp{(i \phi(f))}$}  
of the field, smeared with real test functions $f$ (Weyl operators). 
The resulting algebraic structures can be described as follows. 
Let $\Dc(\RR^4)$ be the space of real valued test functions
with compact support in~$\RR^4$ and let 
$\Dbox(\RR^4) \doteq \Dc(\RR^4) / \square \, \Dc(\RR^4)$ 
be its quotient with regard to test functions lying in the
kernel of the field (that fulfils the wave equation $\square \phi = 0$,
$\square$ being the d'Alembertian). 
We consider the *--algebra 
$\Ac(\RR^4)$, which is generated by all sums and products of the operators
$W(f)$, \mbox{$f \in \Dbox(\RR^4)$}, satisfying the Weyl relations
$$
W(f) W(g) = e^{\, - \kappa(f,g)/2} \, W(f + g) \, , \quad 
W(f)^* = W(-f) \, ,
$$
where  \ $\kappa(f,g) \doteq (2 \pi)^{-3} \!
\int d^4p \, \varepsilon(p_0) \delta(p^2) \widehat{f}(-p) \widehat{g}(p)$
is the commutator function of the field 
and $\widehat{f}, \widehat{g}$ are the Fourier transforms of 
(any member of the classes) $ f,g \in \Dbox(\RR^4)$.
We also consider the  subalgebras $\Ac(\Oa) \subset \Ac(\RR^4)$,
associated with double cones $\Oa \subset \RR^4$,  
which are generated by the unitaries $W(f)$, \mbox{$f \in \Dbox(\Oa)$}, 
where $\Dbox(\Oa) \subset \Dbox(\RR^4)$ denotes the 
class of test functions 
containing members, which have support in $\Oa$.  It follows immediately 
from the properties of the commutator function that 
all elements of $\Ac(\Oa_1)$ commute with those of $\Ac(\Oa_2)$ 
if the double cones $\Oa_1$, $\Oa_2$ are spacelike 
separated (Einstein causality) or timelike separated
(Huygens' principle). The extended Poincar\'e group acts on 
the algebra~$\Ac(\RR^4)$ by automorphisms. Since
for any $\Upsilon = (y, \Lambda, \lambda) \in \ProrW$ the linear map 
$f(x) \mapsto f_\Upsilon(x) \doteq \lambda^{-3} f(\Lambda^{-1} \lambda^{-1}(x - y))$
on $\Dc(\RR^4)$ leaves the subspace $\square \Dc(\RR^4)$
and the commutator function invariant, the automorphic action is 
consistently defined by 
$\alpha_\Upsilon(W(f)) \doteq W(f_\Upsilon)$, $f \in \Dbox(\RR^4)$. It follows 
from this definition that the group acts also covariantly, 
\ie for any double cone $\Oa \subset \RR^4$ one has
$\alpha_\Upsilon(\Ac(\Oa)) = \Ac(\Upsilon \Oa)$, where 
$\Upsilon \, \Oa \doteq (\Lambda \lambda \Oa + y)$ in an obvious notation.
Finally, the Minkowskian vacuum state $\omega_0$
on $\Ac(\RR^4)$ is fixed by (linear extension from) the expectation values 
$$ 
\omega_0(W(f)) \doteq e^{-v(f,f)/2} \, , \quad f \in \Dbox(\RR^4) \, ,
$$
where
$v(f,f) \doteq (2 \pi)^{-3} \! 
\int d^4p \, \theta(p_0) \delta(p^2) \widehat{f}(-p) \widehat{f}(p)$
is the two-point function of the free field. The vacuum state is
distinguished by the fact that it is invariant under the action of
the extended Poincar\'e group, $\omega_0 \circ \alpha_\Upsilon = \omega_0$
for $\Upsilon \in \ProrW$, and that it is a ground state 
for the time translations in any Lorentz system. We will also 
consider thermal equilibrium (KMS) states on $\Ac(\RR^4)$.
By the Gelfand-Naimark-Segal (GNS) construction one can recover 
from any such state a concrete representation 
of the algebra on some Hilbert space. In case of the vacuum state 
$\omega_0$ one obtains the familiar Fock representation
of the free field as well as a continuous unitary representation of~$\ProrW$, 
which induces in this representation the automorphic action of this  
group and satisfies the relativistic spectrum condition. 
We will make use of these facts in subsequent sections. 

\section{Accelerated laboratories}
\setcounter{equation}{0}

We turn now to the description of the motion of a laboratory which initially 
occupies the spherical region $\biL_0 = \{x : x_0=0, \ |\bix - \bio| < r \}$ 
of the time $x_0 = 0$ plane in the chosen Lorentz system.
Its center $\bio$ is the position of the observer carrying along a clock. 
As outlined in the introduction, we consider motions  
under the influence of arbitrary (not necessarily constant) 
accelerations. We do this by taking any sufficiently regular timelike 
and future directed world line of the observer as input. 
Disregarding spatial rotations for simplicity, the motion 
of the laboratory is described by the Fermi-Walker transport 
along that curve. There arise then  
the questions up to which magnitude of 
acceleration this idealised model of a laboratory 
is meaningful and whether its motion can be described by the action of 
Poincar\'e transformations. 

The mathematical formulation of this issue proceeds as follows. 
Let $\gamma: t \mapsto \gamma(t)$ be
the world line of the observer, parametrised by his proper time
$t \geq 0$; it satisfies the initial conditions 
$\gamma(0) = (0,\bio)$, $\dot{\gamma}(0) = (1, \binull)$, where the 
dot indicates the derivative with respect to $t$.
The curve is assumed to be timelike, future-directed and 
twice continuously differentiable; $\dot{\gamma}$ denotes its tangent
vector field giving the normalised time direction 
and $a = \ddot{\gamma}$ is the acceleration. 
Now let $v$ be any differentiable vector field 
along $\gamma$. 
Then $v$ is Fermi-Walker transported along $\gamma$
(cf.\ for example \cite{Str}) if and only if 
\begin{align} \label{FW1}
\dot{v} = (\dot{\gamma} \cdot v) \, a - (a \cdot v) \, \dot{\gamma} \,.
\end{align}
If ${v}(0) \cdot \dot{\gamma}(0) = 0$
at the initial point $t=0$, then it follows from  \eqref{FW1}  
that ${v} \cdot \dot{\gamma} = 0$ all along~$\gamma$; for 
$ \frac{d}{dt}(v \cdot \dot{\gamma}) = 
\dot{v} \cdot \dot{\gamma} + v \cdot \ddot{\gamma} 
= -(a \cdot v) + (v \cdot a) = 0$ in view of the 
normalisation $\dot{\gamma}^2 =1 $.  
Similarly one can conclude the following from \eqref{FW1}: 
Two differentiable vector fields $v$ and $w$
that are Fermi-Walker transported along~$\gamma$ 
are mutually metric-orthonormal (\ie  
$v^2 = -1$, $w^2 = -1$, $v \cdot w = 0$) along all of~$\gamma$
if and only if they have that property at some single point of~$\gamma$. 
This implies that, if at the initial point $t=0$ a set of 
pairwise metric-orthonormal vectors $\kappa_j(0)$, $j = 1,2,3$,  
is chosen such that $(\dot{\gamma}{}(0),\kappa_1(0),\kappa_2(0), 
\kappa_{3}(0))$ forms a four-dimensional Lorentz frame
(tetrad) affixed at $\gamma(0)$, and 
if $\kappa_j(t)$ denotes the vectors
obtained from Fermi-Walker transporting the $\kappa_j(0)$ along $\gamma$ 
from $\gamma(0)$ to $\gamma(t)$,  $j = 1,2,3$, 
then $(\dot{\gamma}{}(t),\kappa_1(t),\kappa_2(t),\kappa_{3}(t))$
is again a four-dimensional Lorentz frame, affixed at~$\gamma(t)$. 
As a matter of fact, the frame obtained by Fermi-Walker 
transport along $\gamma$ is differentiable. Moreover, it is not
difficult to see that there is a continuous 
function of proper orthochronous 
Lorentz transformations $t \mapsto \Lambda(t) \in \ProrL$ 
inducing the Fermi-Walker transport along  $\gamma$, \ie it 
transforms the frame $(\dot{\gamma}{}(0),\kappa_1(0),\kappa_2(0),
\kappa_{3}(0))$ onto the frame $(\dot{\gamma}{}(t),\kappa_1(t),
\kappa_2(t), \kappa_{3}(t))$, $t \geq 0$.

Next, we consider the motion of the laboratory, which is initially 
at rest in the spherical region  
$\biL_0 = \{ x : x_0 = 0, \, |\bix - \bio| < r \} $. As explained,  
we want to describe its motion by the Fermi-Walker transport of $\biL_0$
along the world line $\gamma$. Thus, at proper time $t > 0$ of 
the observer, we are led to assign to the laboratory the region
\begin{equation}  \label{Lotra}
\biL_t = 
\{\gamma(t) + \sum_{j=1}^3 y_j \kappa_j(t) : |\biy| < r \} 
= \gamma(t) + \Lambda(t) (\biL_0 - \gamma(0)) \, ,
\end{equation}
where $\biy = (y_1,y_2,y_3)$ are the spatial coordinates 
relative to the position of the observer in his 
current Lorentz system. This corresponds to the conceived 
experimental situation, where the laboratory is being rigidly 
dragged along the worldline $\gamma$, keeping its 
spherical shape relative to the observer. 

If the laboratory undergoes an accelerated motion, \ie if  
$a = \ddot{\gamma} \neq 0$, 
the various points in the laboratory follow worldlines with 
acceleration other than $a$ and
so there are tidal forces acting on the laboratory. 
Apart from the restrictions on the
acceleration of $\gamma$ stemming from preserving rigidity of 
a realistic laboratory against the tidal forces, there
is an a priory restriction on the acceleration that derives from 
the requirement that any point in the laboratory must follow a 
timelike, future-directed worldline. In other words,
the laboratory is to be represented by a congruence of timelike, 
future-directed worldlines.
To see what this restriction amounts to, let us parametrise 
the worldlines of points of the laboratory according to 
$$
\gamma_{{\bisy}}(t) \doteq 
\gamma(t) + \sum_{j=1}^3 y_j \kappa_j(t) \, , \quad |\biy| < r \, .
$$
Since the $\kappa_j$ are Fermi-Walker transported along 
$\gamma = \gamma_{{\bisnull}}$, Eqn.\ \eqref{FW1} implies
\begin{equation} \label{a-constraint}
\dot{\gamma}_{\bisy}(t)  = 
\dot{\gamma}(t) + \sum_{j=1}^3 y_j \dot{\kappa}_j(t)
 = \big( 1 -  \sum_{j=1}^3 y_j \, a(t)  \cdot \kappa_j(t) 
\big) \, \dot{\gamma}(t) \, , 
\end{equation}
hence all the $\gamma_{\bisy}$ are timelike and future-directed if and 
only if $ \sum_{j=1}^3 y_j \, a(t) \! \cdot \! \kappa_j(t) < 1$
for all $|\biy| < r$ at all times $t \geq 0$.
In view of  $\dot{\gamma} \! \cdot \! a = 0$ and as 
the $\kappa_j(t)$, $j=1,2,3$, form an orthonormal basis of the
hyperplane that lies metric-orthogonal 
to $\dot{\gamma}(t)$, this implies that 
the acceleration must satisfy the condition   
$r \sqrt{-a^2} <1 \, \widehat{=} \, c^2$ all along 
any admissible worldline~$\gamma$.

It may be worth inserting some numbers into this condition
to see that realistic laboratories are quite far from that bound. 
According to references compiled in 
\cite{Wi}, among the highest accelerations known so far are 
$1.9 \cdot 10^9 \, m/s^2$ for protons in the LHC, 
$7 \cdot 10^{12}\,m/s^2$ at the surface of a neutron star, 
and $8.8 \cdot 10^{13}\, m/s^2$ for protons at the Fermilab accelerator. 
As $c^2 = 9 \cdot 10^{16} \, m^2/s^2$, a  perfectly 
rigid laboratory could still extend
about $r = 10^3\,m$ without violating the above bound 
even at these extreme accelerations. Material stresses
would of course deform or destroy a realistic laboratory at much 
lower acceleration scales. The present limit for instrumentation 
acceleration is at about $1.6 \cdot 10^5 \,m/s^2$. 

So the conclusion of this discussion is the insight that 
it is meaningful (i)  to consider spatially extended rigid laboratories 
undergoing acceleration with time scale fixed  by 
the clock of the observer and (ii) to describe the motion of the
laboratory by the action of Poincar\'e transformations, which are 
parametrised by the proper time of the observer, cf.\ 
Eqn.\ \eqref{Lotra}. 

\section{Macroscopic stability of the vacuum 
against acceleration}
\setcounter{equation}{0}  

Having clarified the description of moving laboratories, we will 
show now that the macroscopic properties of the vacuum state 
found by an observer do not change under the influence of 
arbitrary accelerations along his world line $\gamma$. 
The observables in the laboratory region $\biL_t$ 
at his proper time $t \geq 0$
are described by elements of the algebra  $\Ac(\Oa_t)$, where 
the double cone $\Oa_t$ is the causal completion of the  
spatial region~$\biL_t$. 
Note that the indexing of the algebras by double 
cones is merely a matter of notational convenience. It is 
justified by the fact that the underlying Weyl operators 
are defined on the quotient 
$\Dc(\Oa_t) / \square \Dc(\RR^4)$, which corresponds 
to the space of Cauchy data of the wave equation 
with support in~$\biL_t$. 

According to Eqn.\ \eqref{Lotra} one has $\biL_t = \Gamma(t) \, \biL_0$,
where $\Gamma(t) \doteq  
(\gamma(t) - \Lambda(t) \gamma(0), \Lambda(t), 1) \in 
\ProrW$. This entails the corresponding relations for the causal completions,
$\Oa_t =  \Gamma(t) \, \Oa_0$, $t \geq 0$.
It then follows from the covariant action
of the group $\ProrW$ on the algebras that
$\Ac(\Oa_t) = \Ac(\Gamma(t) \, \Oa_0) =
\alpha_{\Gamma(t)}(\Ac(\Oa_0))$, $t \geq 0$. Thus for all 
admissible world lines 
$\gamma$ satisfying the constraint on the acceleration 
given in the preceding section,
the motion of the observables in the laboratory can be described by 
the automorphic action of the group. Being interested in  
persistent macroscopic properties, the observer 
will be led to analyse the states with his observables $\alpha_{\Gamma(t)} (A)$ 
at asymptotic times $t \rightarrow \infty$. As a matter of fact, these   
observables form central sequences whose limits can thus be interpreted
as classical observables.
For the proof we need the following two lemmas,  
the first one being of geometric nature.

\begin{Lem}  \label{geometry}
Let $t \mapsto \gamma(t)$ be any admissible world line of an observer
and let  $\Oa_t$, $t \geq 0$, be the causal completion of 
his Fermi-Walker transported laboratory region of radius $r$, 
where $r \sqrt{-a^2} \leq v^2 < 1$ all along~$\gamma$. 
There exist  an open double cone $\, \Oa \subset \RR^4$ 
and a lightlike vector $l$
such that for any given bounded region 
$\Ba \subset (\Oa + \RR \, l ) \doteq \bigcup_{u \in \RR} \, (\Oa + u \, l )
\subset \RR^4 $
the region $\Oa_t$ is timelike separated 
from $\Ba$ for sufficiently large~$t > 0$.
\end{Lem}
\begin{proof}
The lower tip of the double cone $\Oa_t$ moves along the world line
$t \mapsto \eta(t) \doteq (\gamma(t) - r \dot{\gamma}(t))$, hence
$\Oa_t \subset (V_++ \eta(t))$, where $V_+$ denotes the
forward lightcone, $t \geq 0$. It therefore suffices to establish
the existence of $\Oa$ and $\RR \, l$, as described in the statement,
by replacing $\Oa_t$ with $\eta(t)$, $t \geq 0$. 
Shifting the origin in~$\RR^4$, one may also assume that $\eta(0) = 0$.  
Since $\dot{\eta} (t) = (\dot{\gamma}(t) - r a(t))$
and  $\dot{\eta} (t)^2 = (1 + r^2 a(t)^2) \geq (1 - v^2) > 0$
it follows that $\eta(t) \in V_+$ and
$\eta(t)^2 \geq  (1 - v^2) \, t^2$, $t \geq 0$. 
Depending on the characteristics of the world
line~$\eta$ one must distinguish two cases by means of the 
characteristic hyperplanes 
$H_s(\bis) = \{ x : x_0 - \bis \bix = s \}$
for arbitrary unit vectors $\bis \in S^2$ and times $s \geq 0$. 

The first type of world line $t \mapsto \eta(t)$ has the property 
that it crosses
all of these hyperplanes in the course of time. Let $t_s(\bis)$
be the time where $\eta(t_s(\bis)) \in H_s(\bis) $. Note that 
$t_s(\bis)$ is unique and that for $t > t_s(\bis)$ one has
$\eta(t) \in F_s(\bis) \doteq \{ x : x_0 - \bis \bix > s \}$,
the future of $H_s(\bis)$, since the worldline is timelike
and future directed and the hyperplanes do not contain any 
timelike directions. 
Moreover, for fixed $s \geq 0$, the map $\bis \mapsto t_s(\bis)$ is 
continuous since the world line $\eta$ is continuous. Hence there exists 
the supremum
$t_s \doteq \sup_{\biss \in S^2} \, t_s(\bis)$  since $S^2$ is compact.
Thus for $t > t_s$  one has 
$\eta(t) \in \bigcap_{\biss \in S^2} F_s(\bis) = (V_+ + (s, \binull))$.
Since $s \geq 0$ was arbitrary and since for any given bounded region
$\Ba \subset \RR^4$  one has $ \Ba \subset (-V_+ +(s, \binull))$
for sufficiently large $s \geq 0$ it follows that $\eta(t)$
and therefore also~$\Oa_t$ is timelike separated from $\Ba$ for $t > t_s$.

The second type of world line  $t \mapsto \eta(t)$  has the property 
that there is some hyperplane $ H_{s_0}(\bis_0)$ which it does not cross.
Hence $\eta(t) \in V_+ \bigcap P_{s_0}(\bis_0)$, $t \geq 0$,
where $P_{s_0}(\bis_0) \doteq \{ x : x_0 - \bis_0 \bix < s_0 \}$ is the 
past of $ H_{s_0}(\bis_0)$. One then obtains for the lightlike vector
$l \doteq (1,\bis_0)$ and $u \in \RR$ 
$$
(\eta(t) - u \, l)^2 = \eta(t)^2 - 2 u (\eta_0(t) - \bis_0 \bieta(t))
\geq (1-v^2)t^2 - 2|u| s_0 \, , \quad t \geq 0 \, .
$$
It follows from this estimate 
that for any double cone $\Oa \subset ( -V_+) $ and 
$u$ varying in any bounded interval $\II \subset \RR$, the regions 
$(\Oa + u \, l)$, $u \in \II$, are timelike separated 
from $\eta(t)$ for sufficiently large $t \geq 0$. 
This completes the proof of the statement.
\end{proof}

In the second lemma we show that the algebra
of observables associated with the lightlike cylinder 
$(\Oa + \RR \, l)$, defined in the preceding statement, is irreducibly
represented in the GNS representation 
$(\pi_0, \Hc_0, \Omega_0)$ induced by the inertial vacuum state $\omega_0$. 
Here $\Hc_0$ denotes the familiar Fock space, $\Omega_0 \in \Hc_0$
the Fock vacuum and $\pi_0$ the homomorphism mapping the elements
of the algebra $\Ac(\RR^4)$ to bounded operators in $\Bc(\Hc_0)$.
One then has the equality of expectation
values $\omega_0(A) = \langle \Omega_0, 
\pi_0(A) \Omega_0 \rangle$, $A \in \Ac(\RR^4)$.
We also recall that there is a continuous unitary representation
$U_0$ of the extended Poinar\'e group on~$\Hc_0$ satisfying
$U_0(\Upsilon) \pi_0(A) U_0(\Upsilon)^{-1} =
\pi_0(\alpha_\Upsilon(A))$, $A \in \Ac(\Oa)$
and $U_0(\Upsilon) \Omega_0 = \Omega_0$, $\Upsilon \in \ProrW$. Moreover,
the joint spectrum of the generators of the spacetime translations
$U_0 \upharpoonright \RR^4$ is contained in the closed 
forward lightcone. The following statement, whose proof is given for 
completeness,  is a well-known consequence of these properties.

\begin{Lem} \label{irreducible}
Let $\Oa \subset \RR^4$, $l \in  \RR^4$, be any 
double cone and lightlike translation. respectively, and let
$\Ac(\Oa + \RR \, l)$ be the algebra generated
by $\Ac(\Oa + u \, l)$, $u \in \RR$. This algebra
is irreducibly represented in the vacuum representation
$(\pi_0, \Hc_0, \Omega_0)$, \viz the commutant of 
$\pi_0(\Ac(\Oa + \RR \, l))$ in $\Bc(\Hc)$ consists
of multiples of the identity, 
$\pi_0(\Ac(\Oa + \RR \, l))^\prime = \CC \, 1$.
\end{Lem}  
\begin{proof}
Let $U_0(u l)$, $u \in \RR$, be the unitaries implementing
the subgroup of lightlike translations $\RR \, l$. Due to covariance, 
$U_0(u l) \pi_0(\Ac(\Oa + \RR \, l)) U_0(u l)^{-1}
=  \pi_0(\alpha_{u l}(\Ac(\Oa + \RR \, l))) 
= \pi_0(\Ac(\Oa + \RR \, l))$,  $u \in \RR$,
and, by the Reeh-Schlieder property of the vacuum \cite[Sec.\ II.5.3]{Haag}, 
$\pi_0(\Ac(\Oa + \RR \, l)) \, \Omega_0 \subset \Hc_0$ 
is a dense subspace of~$\Hc_0$.
Now let $Z \in \pi_0(\Ac(\Oa + \RR \, l))^\prime$, then
for any $A \in \Ac(\Oa + \RR \, l)$ and $u \in \RR$
$$
\langle \Omega_0, Z U_0(u l) \pi_0(A) \Omega_0 \rangle
= \langle \Omega_0, Z \pi_0(\alpha_{ul}(A)) \Omega_0 \rangle
= \langle \Omega_0, \pi_0(\alpha_{ul}(A)) Z \Omega_0 \rangle
= \langle \Omega_0, \pi_0(A) U_0(u l)^{-1} Z \Omega_0 \rangle \, ,
$$
as a consequence of covariance and the invariance of $\Omega_0$
under the action of $U_0(u l)$, $u \in \RR$. The 
unitary group of lightlike translations 
$u \mapsto U_0(u l)$, $u \in \RR$, has a positive 
generator and the ray of $\Omega_0$ is the unique invariant subspace
of $\Hc_0$ under its action. (The latter fact is a 
well-known consequence of the 
representation theory of the Poincar\'e group in 
four spacetime dimensions.) It therefore follows from the preceding
equality by Fourier analysis with regard to $u \in \RR$ that
the  matrix elements do not dependent on $u \in \RR$, hence 
$\langle \Omega_0, \pi_0(A) Z \Omega_0 \rangle 
= \langle \Omega_0, \pi_0(A) \Omega_0 \rangle 
\langle \Omega_0, Z \Omega_0 \rangle $, $A \in \Ac(\Oa + \RR \, l)$. 
The Reeh-Schlieder property of $\Omega_0$ then implies
$Z = \langle \Omega_0, Z \Omega_0 \rangle 1$, completing the proof.  
\end{proof}

We are now in the position to prove the main result of this section.

\begin{Prop}
Let $t \mapsto \gamma(t)$ be any admissible world line as in 
Lemma \ref{geometry} and let \
\mbox{$t \mapsto \Gamma(t) \in \ProrW$} be the 
family of Poincar\'e transformations inducing the Fermi-Walker transport
of the laboratory region, \mbox{$\Oa_t = \Gamma(t) \, \Oa_0$, $t \geq 0$}. 
There exist in the vacuum representation 
$(\pi_0, \Hc_0, \Omega_0)$ of $\Ac(\RR^4)$ the limits
$$
\lim_{t \rightarrow \infty} \, \pi_0(\alpha_{\Gamma(t)}(A)) 
= \omega_0(A) \, 1 \, , \quad A \in \Ac(\Oa_0 ) \, ,
$$
in the weak operator topology.
\end{Prop}
\begin{proof}
Given any $A \in \Ac(\Oa_0)$ 
one has $\alpha_{\Gamma(t)}(A) \in \Ac(\Gamma(t) \Oa) =  \Ac(\Oa_t)$, $t \geq 0$.
Now according to Lemma~\ref{geometry} there exist a double cone $\Oa$
and a lightlike vector $l \in \RR^4$ such that for any
finite interval \mbox{$\II \subset \RR$} the regions
$\Oa_t$ and $(\Oa + \II \, l)$  are timelike separated for
sufficiently large $t \geq 0$. It therefore follows from Huygens'
principle that the operators $\alpha_{\Gamma(t)}(A)$ commute with
any given operator \mbox{$B \in \Ac(\Oa + \RR \, l)$} in the limit
$t \rightarrow \infty$. 

Proceeding to the vacuum representation 
one makes use of the fact that 
for any given $A \in \Ac(\Oa_0)$ the family of operators 
$ \pi_0(\alpha_{\Gamma(t)}(A)) $, $t \geq 0$, is uniformly bounded and
thus has weak limit points in~$\Bc(\Hc_0)$. Let $t_n \geq 0$, $n \in \NN$,
be any sequence such that there exists the limit 
$Z \doteq \lim_{n \rightarrow \infty} \pi_0(\alpha_{\Gamma(t_n)}(A))$
in the weak operator topology. As a consequence of Huygens' principle one 
has $[Z, \pi_0(B)] = 0$ for any $B \in  \Ac(\Oa + \RR \, l)$ and
since $\pi_0( \Ac(\Oa + \RR \, l))$ is irreducible, cf.\   
Lemma~\ref{irreducible}, it follows that $Z = z 1$ for some $z \in \CC$. 
Now 
\begin{equation*}
\begin{split}
z  = \langle \Omega_0, Z \Omega_0 \rangle  =
\lim_{t \rightarrow \infty}   \langle \Omega_0, \pi_0(\alpha_{\Gamma(t_n)}(A)) 
\Omega_0 \rangle =
\lim_{t \rightarrow \infty}  \omega_0(\alpha_{\Gamma(t_n)}(A))
= \omega_0(A) \, ,
\end{split}
\end{equation*}
where the last equality obtains from the invariance of the vacuum under 
the action of the Poincar\'e group. Thus all weak limit points of
$ \pi_0(\alpha_{\Gamma(t)}(A)) $ \ for $t \rightarrow \infty$ coincide 
with $\omega_0(A) \, 1$, so this family of operators is convergent 
in the weak operator topology for $A \in \Ac(\Oa_0)$, as stated.  
\end{proof}

This result shows that quantum effects, leading  
to excitations of the vacuum described by vectors in $\Hc_0$, are 
suppressed in the limit of large times: the observables 
$\alpha_{\Gamma(t)}(A)$, $t \geq 0$, form central sequences 
in the vacuum representation which 
converge weakly and have sharp (non-fluctuating) limits 
due to the asymptotic abelianess of the dynamics. 
As a matter of fact, by taking suitable time averages 
one can also accomplish their convergence in the strong operator 
topology. The  
limits of these sequences give information about 
persistent macroscopic properties of the underlying states. 
In view of the fact that they
do not depend 
on the world line $\gamma$, we conclude that an 
accelerated observer will not
register any macroscopic effects of the acceleration in the vacuum state. 
In particular, he does not 
experience the vacuum state as a thermal gas, respectively bath,
contrary to statements made in the literature, cf.\ \cite[p.\ 167]{FuRu},
\cite[p.\ 3721]{Un1}, \cite[p.\ 115]{Wald}.
For such a gas would leave a trace in the asymptotic  
expectation values of some observable $A$ which differs 
from its vacuum expectation value~$\omega_0(A)$. 
 
We conclude this section by indicating how these central 
sequences of observables may be used in order to determine the temperature of
states, thereby complementing the discussion in \cite{BuSo}.
Within the present model, the inertial equilibrium states 
in the chosen Lorentz system are described 
by functionals $\omega_T$, $T > 0$, on 
$\Ac(\RR^4)$. They are fixed by (linear extension from) 
the expectation values
$$ 
\omega_T(W(f)) \doteq e^{-v_T(f,f)/2} \, , \quad f \in \Dbox(\RR^4) \, ,
$$
where
$v_T(f,f) \doteq (2 \pi)^{-3} \! 
\int d^4p \, \varepsilon(p_0) \delta(p^2) (1 - e^{-p_0/T})^{-1} 
\widehat{f}(-p) \widehat{f}(p)$
is the thermal two-point function of the free field. 
Proceeding to the corresponding GNS-representations 
$(\pi_T, \Hc_T, \Omega_T)$ one can show that for 
inertial motion $t \mapsto \gamma(t) \doteq (t, \binull)$ 
one has 
$\lim_{t \rightarrow \infty} \pi_T(\alpha_{\gamma(t)}(A)) = \omega_T(A) \, 1$,
$A \in \RR^4$, in the weak operator topology. 
(This is a well known consequence of the fact that the KMS states 
$\omega_T$ are faithful and have the mixing property, 
\ie they describe pure phases.)
For fixed~$A$ the mapping $T \mapsto \omega_T(A)$
describes the equation of state of the intensive 
quantity measured by $A$ as a function of the 
inertial temperature. Thus the observer may 
calibrate his observables in the inertial equilibrium states
and use them as empirical thermometers also when he is in 
motion. According to the preceding results he would then come
to the conclusion that in the presence of acceleration
the temperature of the vacuum remains to be~zero,
in accordance with the results in \cite{BuSo}.

\section{Macroscopic properties of accelerated 
equilibrium states}
\setcounter{equation}{0}

If an observer undergoes some constant 
acceleration of modulus $\sfa$ he will be able to
prepare in his laboratory equilibrium states and study their macroscopic
properties. In this section we determine the predictions which     
the present model makes about his findings. It will turn 
out that, in contrast to the inertial situation, the 
global parameter $T_\sfa$ characterising the accelerated equilibrium
states cannot be interpreted as temperature anymore.
In fact, disregarding the inertial vacuum,
the temperature varies spatially within these states.

In order to keep the notation simple, we 
assume that the observer moves with constant acceleration $\sfa > 0$
along the world line
$t \mapsto \gamma_\sfa(t) = (sh(t\sfa)/\sfa,ch(t\sfa)/\sfa,0,0)$, $t \geq 0$.
As is well-known and can also be inferred from the discussion in Sec.~2, 
the maximally possible
initial laboratory region is, in this special case, the half space 
$\biL_0 = \{x : x_0=0, \, x_1 > 0 \}$. At later times this region 
is Fermi-Walker transported to 
$\biL_t = \{ x : x_0 = th(t\sfa) \, x_1 > 0 \} = \Lambda_\sfa(t) \, \biL_0$, 
where $\Lambda_\sfa(t) \in \ProrL$ are the boosts 
fixed by the given worldline, $t \geq 0$. The causal completion of each  
of these laboratory regions is the wedge shaped region \
$\Wa \doteq \{ x : x_1 > |x_0| \} \subset \Mc$ which is 
stable under the action of the boosts. 
Let us mention as an aside that one should not 
think of this wedge region as Rindler space, having an insurmountable 
boundary at its edge (the apparent horizon). For that 
idea might induce one to take the Rindler vacuum as 
a (global) reference state, which is disjoint from the 
inertial vacuum and leads  
to an interpretation of the theory in terms of ``Rindler quanta''.  
Such an interpretation would not correspond to the situation
treated here, where a Minkowski space based observer 
enters a spacecraft in order to perform experiments
with his local observables. 

The equilibrium states which, in principle, can be prepared by 
the accelerated observer in his laboratory
are stationary and passive \cite{PuWo} and hence are also 
described by KMS states $\omega_{T_\sfa}$, $T_\sfa > 0$, 
on the algebra $\Ac(\Wc)$ with regard to the automorphic action
$\alpha_{\Lambda_\sfa(t)}$, $t \in \RR$, of the  
time translations. Note that the index $\sfa$ 
will be used throughout in order to distinguish quantities related to the
latter dynamics. 
The KMS states are  fixed by (linear extension from) 
the expectation values
\begin{equation} \label{kms1}
\omega_{T_\sfa}(W(f)) \doteq e^{-v_{T_\sfa}(f,f)/2} \, , \quad f \in \Dbox(\Wc) \, ,
\end{equation}
where the thermal two-point function $v_{T_\sfa}$ is now given by 
\begin{equation} \label{kms2}
v_{T_\sfa}(f,f) \doteq (2 \pi)^{-4} \! 
\int \! du \, (1 - e^{-u/{T_\sfa}})^{-1} 
\int \! dv \, e^{-iuv} \!   
\int d^4p \, \varepsilon(p_0) \delta(p^2) \, 
\widehat{f}(-p) \widehat{f}(\Lambda_\sfa(v) \, p) \, . 
\end{equation}
Since we are not aware of a reference where it has been shown 
that the functionals $\omega_{T_\sfa}$, $T_\sfa > 0$,  
are KMS states (satisfying the condition of positivity)
in case of the free massless scalar field 
we provide a proof in an appendix. 
The well-known fact that for 
the special value $T_\sfa = \sfa/2 \pi$ the corresponding 
state~$\omega_{T_\sfa}$
coincides with the restriction of the inertial vacuum to the
wedge algebra, $\omega_{T_\sfa} = \omega_0 \upharpoonright \Ac(\Wc)$, 
is commonly interpreted as formal evidence for the Unruh effect
\cite{Fulling, Sewell, UnWa}. 

In analogy to the results established in 
the preceding section one can show that  
observables that are localised in relatively compact regions 
of $\Wc$ form central sequences at asymptotic times, 
which have sharp limit values in the GNS-representations 
induced by the 
KMS states $\omega_{T_a}$ on  $\Ac(\Wc)$, $T_a > 0$. 
The proof is based on standard arguments and 
given here for completeness. 

\begin{Prop} \label{centrallimits} 
Let $\, \omega_{T_\sfa}$, $T_\sfa > 0$, be any 
KMS state, defined above relative to the 
automorphic action of the dynamics $\alpha_{\Lambda_\sfa(t)}$, $t \in \RR$,
on the algebra $\Ac(\Wc)$, and let
$(\pi_{T_\sfa}, \Hc_{T_\sfa}, \Omega_{T_\sfa})$ be the 
corresponding GNS-representation. There exist the limits 
$$
\lim_{t \rightarrow \infty} \, \pi_{T_\sfa}(\alpha_{\Lambda_\sfa(t)}(A))
= \omega_{T_\sfa}(A) \, 1 \, , \quad A \in \Ac(\Wc) \, ,
$$
in the weak operator topology. 
\end{Prop}
\begin{proof}
The crucial step in the argument is the proof that the state
$\omega_{T_\sfa}$ is mixing, \ie for any pair of operators 
$A,B \in \Ac(\Wc)$ one has \ 
$\lim_{t \rightarrow \infty} \, \omega_{T_\sfa}(B \, \alpha_{\Lambda_\sfa(t)}(A))
=  \omega_{T_\sfa}(B) \, \omega_{T_\sfa}(A)$. 
This property implies that 
$\lim_{\, t \rightarrow \infty} \pi_{T_\sfa}( \alpha_{\Lambda_\sfa(t)}(A)) \, 
\Omega_{T_\sfa} = \omega_{T_\sfa}(A) \, \Omega_{T_\sfa}$, $A \in \Ac(\Wc)$, 
in the sense of weak convergence in $\Hc_{T_\sfa}$.
Since the family of operators $\pi_{T_\sfa}( \alpha_{\Lambda_\sfa(t)}(A))$, 
$t \geq 0$, is uniformly bounded, hence has weak limit points, and 
KMS states are separating
for the weak closure of the represented algebra, the statement 
then follows. For the proof of the mixing property one 
makes use of the fact that the elements of $\Ac(\Wc)$ are linear 
combinations of Weyl operators, so it suffices to
consider the functions $t \mapsto 
\omega_{T_\sfa} (W(g) \alpha_{\Lambda_\sfa (t)}(W(f)))$, $f,g \in \Dbox(\Wc)$.
Applying the Weyl relations one gets   
$$
\omega_{T_\sfa} (W(g) \alpha_{\Lambda_\sfa(t)}(W(f))) =
e^{- \kappa(g,f_{\Lambda_\sfa (t)^{-1}})/2} \,
\omega_{T_\sfa} (W(g + f_{\Lambda_\sfa(t)^{-1}}))
= e^{- v_{T_\sfa}(g,f_{\Lambda_\sfa (t)^{-1}})} \,
e^{-v_{T_\sfa}(g,g)/2} \,
e^{-v_{T_\sfa}(f,f)/2} \, .
$$
Furthermore, making use of the group law 
$\Lambda_\sfa(v) \Lambda_\sfa(t) = \Lambda_\sfa(v+t)$, one obtains 
$$
v_{T_\sfa}(g,f_{\Lambda_\sfa (t)^{-1}}) = (2 \pi)^{-4} \! 
\int \! du \, e^{itu} (1 - e^{-u/{T_\sfa}})^{-1} 
\int \! dv \, e^{-iuv} \!   
\int d^4p \, \varepsilon(p_0) \delta(p^2) \, 
\widehat{g}(-p) \widehat{f}(\Lambda_\sfa(v) \, p) \, . 
$$
It is shown in the appendix that \
$u \mapsto \int \! dv \, e^{-iuv} \!   
\int d^4p \, \varepsilon(p_0) \delta(p^2) \, 
\widehat{g}(-p) \widehat{f}(\Lambda_\sfa(v) \, p)$ is
a test function that vanishes at $u = 0$ for any choice of 
$f,g \in \Dbox(\Wc)$. The Riemann-Lebesgue lemma then  
implies \ $\lim_{\, t \rightarrow \infty}  v_{T_\sfa}(g,f_{\Lambda^{-1}_\sfa(t)}) = 0$
and hence \ $\lim_{\, t \rightarrow \infty} \, 
\omega_{T_\sfa} (W(g) \alpha_{\Lambda_\sfa (t)}(W(f))) =
\omega_{T_\sfa} (W(g)) \, \omega_{T_\sfa} (W(f))$, completing the proof.  
\end{proof}

Thus, in the presence of constant acceleration, an observer 
can still describe the macroscopic properties of the  
equilibrium states in his laboratory in terms of relations  
$T_\sfa \mapsto \omega_{T_\sfa}(A)$ involving his intensive observables,  
which are parametrised by the equilibrium parameter 
$T_\sfa > 0$. In contrast to the inertial situation, the
equilibrium states are not spatially homogeneous, however.  
Note that the semigroup of spatial translations 
$\RR^3_+ \doteq \{ \yb : y_1 \geq 0 \}$ maps the wedge $\Wc$ into itself. 

\begin{Prop}
Let $\, \omega_{T_\sfa}$, $T_\sfa > 0$, be any KMS state, defined above
relative to the automorphic action of the dynamics 
$\alpha_{\Lambda_\sfa(t)}$, $t \in \RR$  on the algebra $\Ac(\Wc)$.  
If $\omega_{T_\sfa}$ is invariant
under the semigroup of spatial translations, \viz 
$\omega_{T_\sfa} \circ\alpha_\yb =  \omega_{T_\sfa}$, $\yb \in \RR^3_+$, \   
then $T_\sfa = \sfa/2\pi$ and $\omega_{T_\sfa}$ coincides with 
the restriction of the inertial vacuum state $\omega_0$ 
to $\Ac(\Wc)$. 
\end{Prop}
\begin{proof}
The boosts $\Bc_\sfa \doteq \{ \Lambda_\sfa(t) : t \in \RR \}$ 
and spatial translations
$\RR^3_+ = \{\biy : y_1 \geq 0 \}$ generate the semigroup
$\overline{\Wc} \rtimes \Bc_\sfa$, 
where $\overline{\Wc} = \{y : y_1 \geq |y_0| \geq 0 \} \subset \RR^4$
denotes the cone of spacetime translations acting as 
endomorphisms on $\Wc$. 
Since $\omega_{T_\sfa}$ is invariant under the automorphic actions of the  
boosts and spatial translations, it is also invariant under the 
action of the semigroup, \ie 
$\omega_{T_\sfa} \circ \alpha_{\Sigma} = \omega_{T_\sfa}$, 
\mbox{$\Sigma \in \overline{\Wc} \rtimes \Bc_\sfa$}. 
Proceeding to the GNS-representation
$(\pi_{T_\sfa},\Hc_{T_\sfa},\Omega_{T_\sfa})$ induced by $\omega_{T_\sfa}$,
this implies that there exists a continuous unitary representation \
$U_{T_\sfa} : \overline{\Wc} \rtimes \Bc_\sfa 
\rightarrow \Bc(\Hc_{T_\sfa})$ \ given by \ 
$U_{T_\sfa}(\Sigma) \, \pi_{T_\sfa}(A) \Omega_{T_\sfa} 
\doteq \pi_{T_\sfa}(\alpha_{\Sigma}(A)) \Omega_{T_\sfa}$,
$A \in \Ac(\Wc)$. (The isometries $U_{T_\sfa}(\Sigma)$
have a dense range due to the Reeh-Schlieder property 
of $\Omega_{T_\sfa}$, cf.~\cite{StVeWo}.) Since the group 
$(\overline{\Wc} - \overline{\Wc}) = \RR^4$ is abelian
and the boosts~$\Bc_\sfa$ normalise $ \overline{\Wc}$,   
one can consistently extend the unitary representation $U_\sfa$ of
\ $\overline{\Wc} \rtimes \Bc_\sfa$  
to \ $\RR^4 \rtimes \Bc_\sfa$, 
putting $U_\sfa(x - y, \Lambda_\sfa) \doteq U_\sfa(x) U_\sfa(y)^{-1} 
U_\sfa(\Lambda_\sfa)$ for $x,y \in \overline{\Wc}$, $\Lambda_\sfa \in \Bc_\sfa$. 

Consider now the lightlike translations $l_\pm =  (\pm l, l,0,0)$, $l \in \RR$, 
satisfying  $\Lambda_\sfa(t) \, l_\pm = e^{\pm ta} \, l_\pm $, $t \in \RR$.
Thus 
$\alpha_{l_\pm}(\Ac(\Wc)) = \Ac(\Wc + l_\pm) \subset \Ac(\Wc)$
for any given $l>0$ and 
$\alpha_{\Lambda_\sfa(t)}(\Ac(\Wc + l_\pm))
\subset \Ac(\Wc + l_\pm)$ for $\pm t \geq 0$ and $l>0$.
Since $\omega_{T_\sfa}$ is a KMS state, 
$t \rightarrow U_\sfa(\Lambda_\sfa(t))$ is (after rescaling
of $t$) the modular group associated with the pair 
$(\pi_{T_\sfa}(\Ac(\Wc))^-,\Omega_{T_\sfa})$, where the bar ${}^-$ denotes
closure in the weak operator topology. 
Hence for any $l > 0$  the algebras 
$U_\sfa(l_\pm) \pi_{T_\sfa}(\Ac(\Wc))^- U_\sfa(l_\pm)^{-1} \subset 
\pi_{T_\sfa}(\Ac(\Wc))^-$ form ``half-sided modular inclusions'' 
which implies that the one-parameter groups
$l \mapsto U_\sfa(l_\pm)$, $l \in \RR$ have 
positive and negative
generators, respectively, cf.\ \cite{Wiesbrock}. Thus, putting $l = 1/2$, 
the generator of the group of inertial 
time translations $t \mapsto U_\sfa(t(l_+ - l_-))$, $t \in \RR$,  
is positive and $\Omega_{T_\sfa} \in \Hc_{T_\sfa}$ is a ground 
state for it. Since the free massless scalar field has a unique
scale invariant inertial ground state, \viz the 
inertial vacuum $\omega_0$, this shows that 
the given state $\omega_{T_\sfa}$ coincides
with the restriction $\omega_0 \upharpoonright \Ac(\Wc)$,
proving that~$T_\sfa = \sfa/2 \pi$. 
\end{proof}

The preceding result can be established in the present model also 
by explicit computations. Our general argument, however,    
shows that the spatial inhomogeneity of equilibrium states in a 
uniformly accelerated laboratory is a model independent feature. In the 
case at hand, we are dealing with equilibrium states of a 
relativistic ideal gas, where one 
expects a one-to-one correspondence between the local (empirical) 
temperature and the asymptotic expectation values of 
suitable intensive observables~$A$. The fact
that $\bix \mapsto \omega_{T_\sfa}(\alpha_{\bisx}(A))$
is in general not constant 
therefore already suggests that the global parameter $T_\sfa$
cannot be interpreted as temperature. This 
point will be substantiated
in the subsequent discussion, where it is shown that in the present model 
all equilibrium states coincide
at sufficiently large distance from the edge of
the wedge $\Wc$ (the horizon).
For the proof of this assertion we make use of the following  
facts.

\medskip 
\noindent (i) \ All KMS states $\omega_{T_\sfa}$, $T_\sfa > 0$, 
on $\Ac(\Wc)$ are invariant under the automorphic action of the 
dilations $(0,1,\lambda) \in \ProrW$, $\lambda > 0$, \ie
$\omega_{T_\sfa} \circ \alpha_\lambda = \omega_{T_\sfa}$ in short hand
notation. This is easily inferred from the definition
of the states and the fact that the dilations commute with 
Lorentz transformations. 

\medskip 
\noindent (ii) 
The restrictions of the states~$\omega_{T_\sfa} \upharpoonright \Ac(\Oa)$,
$T_\sfa > 0$, to the algebra of any relatively compact region 
$\Oa \subset \Wc$ are normal with respect to each other, \ie
they are continuous on the unit ball of the algebra 
in the weak operator topology induced by any one of these states. 
This follows from the fact that the KMS states $\omega_{T_\sfa}$, 
$T_\sfa > 0$, are quasifree Hadamard states~\cite{SahlmVer} and that
such states are locally normal with respect to each 
other \cite{Ver-QE}. 
The latter fact implies that the  
states can locally be interpreted in terms of ensembles 
described by density matrices in the Fock space of the
inertial vacuum.
But, similarly to the inertial case, different KMS states are not normal 
(in fact, disjoint) with respect to each other on the whole algebra~$\Ac(\Wc)$. 

\medskip 
\noindent (iii) \ Finally,  let $\Oa_{R,r} \subset \Wc$ be the double cone 
centred at $\bio_R = (0,R,0,0)$ with spherical basis of radius $0 < r < R$ 
at time $x_0 = 0$. The dilations act on this region according
to $\lambda \Oa_{R,r} = \Oa_{\lambda R, \lambda r}$, $\lambda > 0$. 
Because of the invariance of the thermal states
under the automorphic action of the 
dilations one obtains for their local norm distance
$ \| \omega_{T_\sfa} - \omega_{T^\prime_\sfa} \|_{\, \Oa_{R,r}} 
\doteq \sup_{A \in \Ac(\Oa_{R,r})} \, 
|\omega_{T_\sfa}(A) - \omega_{T^\prime_\sfa}(A)|/\|A\| $
the equalities 
$ \| \omega_{T_\sfa} - \omega_{T^\prime_\sfa} \|_{\, \Oa_{R,r}} 
= \| \omega_{T_\sfa} \circ \alpha_\lambda 
 - \omega_{T^\prime_\sfa} \circ \alpha_\lambda \|_{\, \Oa_{R,r}} 
= \| \omega_{T_\sfa} - \omega_{T^\prime_\sfa} \|_{\, \Oa_{\lambda R, \lambda r}}$,
$\lambda > 0$.
After these preparations we can establish the following fact, 
which relies on arguments given by Roberts in~\cite{Roberts}.

\begin{Prop}
Let $\omega_{T_\sfa}, \omega_{T^\prime_\sfa}$ be any pair of KMS states 
on the algebra $\Ac(\Wc)$ with regard to the automorphic action of the 
dynamics $\alpha_{\Lambda_\sfa(t)}$, $t \in \RR$.
Then \ $\lim_{\, r/R \rightarrow 0} \
\| \omega_{T_\sfa} - \omega_{T^\prime_\sfa} \|_{\, \Oa_{R,r}} = 0$. 
In particular, the norm distance vanishes in this limit for 
any ${T_\sfa} > 0$ and fixed ${T^\prime_\sfa} = a/2\pi$, where 
$\omega_{T^\prime_\sfa}$ coincides with the restriction of the
inertial vacuum $\omega_0$ to the wedge algebra, 
$\omega_{T^\prime_\sfa} = \omega_0 \upharpoonright \Ac(\Wc)$. 
\end{Prop}
\begin{proof}
Because of the triangle inequality for the norm distance
it suffices to prove the statement for the special case 
\mbox{$\omega_{T^\prime_\sfa} = \omega_0 \upharpoonright \Ac(\Wc)$}. 
Moreover, inserting into the equality 
$ \| \omega_{T_\sfa} - \omega_0 \|_{\, \Oa_{R,r}} 
= \| \omega_{T_\sfa} - \omega_0 \|_{\, \Oa_{\lambda R, \lambda r}}$
the special value $\lambda = 1/R$, one only needs to 
estimate the norm distances 
$\| \omega_{T_\sfa} - \omega_0 \|_{\, \Oa_{1, \, r/R}}$.
Copying the argument in \cite{Roberts}, 
this is accomplished by making 
use of the fact that the intersection
of the double cones $\Oa_{1, \, r/R}$, $r/R > 0$, 
consists of the single point $(0,1,0,0) \in \Wc$.
Hence the intersection of the corresponding weakly closed 
algebras in the GNS representations $(\pi_0. \Hc_0, \Omega_0)$
induced by $\omega_0$
consists of multiples of the identity, 
$\bigcap_{\,  r/R > 0} \, \pi_0(\Ac(\Oa_{1, \, r/R}))^- = \CC \, 1.$ 
The latter relation says that there exist no non-trivial 
bounded operators that are localised at a point, 
a classical result due to Wightman~\cite{Wightman}. 
Now there exists for any $r/R > 0$ an operator
$A_{\, r/R} \in \Ac(\Oa_{1, \, r/R})$ with norm $\| A_{\, r/R} \| = 1$
such that 
$\| \omega_{T_\sfa} - \omega_0 \|_{\, \Oa_{1, \, r/R}} 
\leq \big(|\omega_{T_\sfa}(A_{\, r/R}) - \omega_0(A_{\, r/R})| + (r/R) \big)$.
Since all weak limit points of the uniformly bounded 
sequence of operators 
$\pi_0(A_{\, r/R})$, $r/R \rightarrow 0$, are contained in 
$\bigcap_{\,  r/R > 0} \, \pi_0(\Ac(\Oa_{1, \, r/R}))^-$ and hence are 
multiples of the identity, and since the states~$\omega_{T_\sfa}$
are locally normal with respect to $\omega_0$, 
this implies $\lim_{\, r/R \rightarrow 0} \
\| \omega_{T_\sfa} - \omega_0 \|_{\, \Oa_{1, \, r/R}} = 0$ and the
statement follows.  \end{proof}

This result shows that the states $\omega_{T_\sfa}$, $T_\sfa > 0$,
can practically not be discriminated from the inertial vacuum
$\omega_0$ by observations in regions of arbitrarily large radius $r$ that 
are separated from the edge of the wedge $\Wc$ (the apparent horizon)
by a distance 
$R \gg r$. Thus in spite of the fact that these states correspond to 
different equilibrium parameters~${T_\sfa}$, 
one clearly must assign to them
in these remote regions the same temperature as to the inertial vacuum. 
On the other hand, the state 
$\omega_{\sfa/2 \pi} = \omega_0 \upharpoonright \Ac(\Wc) $ 
is spatially homogeneous with regard to all local observables in the 
laboratory and hence has the same temperature everywhere. 
Since the
Tolman--Ehrenfest law~\cite{To,EhTo} implies that the 
local temperature in a uniformly accelerated equilibrium state 
is proportional to its inverse 
distance from the horizon, the temperature must be zero 
everywhere in state $\omega_{\sfa/2 \pi}$.
This fact substantiates our assertion that the
temperature of the inertial vacuum remains to be zero   
in uniformly accelerated systems and that the global 
equilibrium parameters~$T_\sfa$, even if corrected 
by redshift factors, cannot directly be interpreted as temperature 
of accelerated equilibrium states. 
We refer the reader to \cite{BuSo} for a definition 
of observables indicating the effective local temperature  
of equilibrium states in the present setting, cf.\ also~\cite{BuVe} and the 
subsequent concluding remarks.

\addtolength\textheight{3mm}
\section{Conclusions}

In the present article we have studied the macroscopic 
effects of acceleration on equilibrium states, as seen by an  
observer in a rigid, spatially extended laboratory. 
The macroscopic properties of these states
are determined by local observables in the
respective laboratory system, which form central sequences
at asymptotic times.
These sequences have sharp limits, hence  
quantum fluctuations are suppressed.
It turned out that acceleration does not affect  
the macroscopic properties of an inertial 
vacuum state. Irrespective of the  
accelerated, possibly erratic motion of the
laboratory, the observer will find  
the same macroscopic properties of the vacuum 
as an inertial observer. In particular, he will
not find himself immersed in a thermal gas,
respectively heat bath. 

We have also shown that the equilibrium parameter $T_\sfa$, distinguishing 
the KMS states in a uniformly accelerated laboratory, cannot offhandedly 
be interpreted as temperature. Disregarding the particular value  
\mbox{$T_\sfa = \sfa/2\pi$}, the 
states are inhomogeneous and coincide at sufficiently large distances 
from the horizon of the observer with the inertial vacuum.
Hence, in spite of the fact that these states correspond to different
equilibrium parameters~$T_\sfa$, one must assign to them the same 
temperature in these remote regions.
Moreover, as a consequence of the Tolman-Ehrenfest law, 
the parameter $T_\sfa = \sfa/2\pi$ attributed to the spatially 
homogeneous vacuum may not be regarded as its temperature either;
in fact, the temperature must vanish throughout this state
according to this law.

An operationally meaningful definition of temperature, 
based on the concept of local thermometer observables, was 
proposed in \cite{BuSo}. 
Proceeding to the idealisation of pointlike observables, the 
simplest example of a local thermometer is, in the present setting, 
the normal ordered square of the underlying free field, 
$\Theta(x) \doteq  12 \! : \! \! \phi^2 \! \! : \! \! (x) $.  
The numerical factor is determined by calibration in the inertial 
equilibrium states $\omega_T$, yielding the expectation values \ 
$\omega_T(\Theta(x)) = T^2$, $T \geq 0$.
Taking suitably regularised time limits of this thermometer 
observable in the representation induced by any given 
accelerated KMS state  $\omega_{T_\sfa}$, $T_\sfa > 0$,
one obtains for the local temperature $T_\sfa (x_1)$ at distance $x_1$ from the 
horizon the result~\cite{BuSo}
\begin{equation} \label{additional}
T_\sfa^{\, 2}(x_1) = \omega_{T_\sfa}(\Theta(x_1)) 
=  (\sfa x_1)^{-2} \, (T_\sfa^2 - (\sfa/2 \pi)^2) \, .
\end{equation}
Thus, in accordance with the Tolman--Ehrenfest law,
one finds that $\sfa x_1 \, {T_\sfa} (x_1) = \mbox{const}$, 
$x_1 > 0$, in the given states. Note that the constant 
appearing in this law does not coincide with the KMS 
parameter~$T_\sfa$, contrary to the common {\it ad hoc} 
definition of local temperatures, where the global KMS 
parameters~$T_\sfa$ are divided through the 
local redshift factors $\sfa x_1$. In fact, 
the constant is modified by a contribution due to the 
Unruh temperature, which vanishes only in the 
classical limit. So in this approach  
the temperature attributed to the inertial 
vacuum in the accelerated laboratory turns also out to be zero, 
$T_{\sfa/2 \pi}(x_1) =  0$ for $x_1 > 0$. This result is   
in accordance with the present findings, where we did not 
rely on an {\it a priori} concept of local thermometers. 

Some contingent objection against these observations  
derives from the fact that microscopic probes, which are 
locally coupled to accelerated KMS states in order to model 
``local thermometers'',  
are driven to Gibbs ensembles corresponding to one and the same 
parameter $T_\sfa$ which characterises the respective 
underlying  macroscopic KMS state. This fact is often taken 
as an argument that $T_\sfa$ ought to be 
interpreted as temperature of that state. 
Yet, as has been outlined in the introduction, any local coupling  
does not only induce the transfer of thermal energy (heat) between
the KMS states and the probe; 
it inevitably creates also excitations of these states
because of the quantum nature of the coupling \cite[p 334]{SDC} (KMS states are 
faithful, cf.~\cite{Haag}). These quantum induced  
excitations gain energy by the acceleration that  
is partly transferred to the probes. The probes therefore indicate  
temperatures $T_\sfa$ which are higher than the local (redshifted) 
temperature of the underlying KMS state, 
$T_\sfa > \sfa x_1 T_\sfa(x_1) $, cf. relation \eqref{additional}. 

This mechanism is also effective in the inertial vacuum state. 
There the value $T_\sfa  = \sfa/2\pi$ indicated by the probe is 
entirely due to the energetic contributions of
quantum induced excitations, there are no contributions
coming from a ``vacuum gas''.
Phrased differently, instead of indicating the temperature
of the vacuum, the probe indicates its own 
temperature at asymptotic times, 
which is raised by local quantum effects during 
the measuring process \cite{SDC}. This conclusion is 
confirmed by the observation that the value of~$T_\sfa$
does not depend on the position of the probe within the  
laboratory,  \ie on its particular world line. The latter 
fact can be extracted from remarks in \mbox{\cite[pp 6531-6532]{BiMe}}
about the arbitrary choice of form factors, determining the position of 
the probe within the laboratory. As already noticed by these authors, 
this feature is at variance with the conventional interpretation
of the Unruh effect. 
So the Unruh effect is a quantum induced systematic contribution,  
which appears in certain specific measuring procedures of temperature,  
but, as explained, may be avoided by others~\cite{BuSo}. Its 
popular thermal interpretation is not tenable, however.

The upshot of the present investigation, going beyond the case 
considered here, is the insight that the parameters  $T_\ast$, characterising
KMS states, may in general not directly be interpreted as temperatures 
and that probes (Unruh-de Witt detectors), which reliably determine these 
parameters, may therefore not be regarded as perfect thermometers. In the 
presence of acceleration or, equivalently, gravitation and also of 
curvature, cf. \cite{BuSch,So}, the parameters $T_\ast$ subsume information
about the local temperature of the underlying equilibrium states as well as
of these other local data, cf.\ relation~\eqref{additional}. Phrased 
differently, the value of $T_\ast$ dictates the relation between 
these local parameters which is required to obtain global equilibrium.
For inertial systems or small accelerations and curvatures, $T_\ast$ may 
safely be identified with temperature; to give an example, for terrestrial 
acceleration $\sfa = 9.8 m/s^2$ and local temperature $T = 300 \, {}^0 \! K$
the corresponding KMS parameter $T_\sfa$ indicated by a probe would be,
according to theory, about $2.7 \, 10^{-42} \, {}^0 \! K$ higher. 
But in the neighbourhood of huge masses or black holes these systematic effects  
can no longer be neglected and the interpretation of thermal properties
of states should then no longer be based directly on the  
KMS parameters $T_\ast$. Further investigations of this issue
therefore seem warranted. 

\addtolength\textheight{-3mm}
\section*{Appendix}

We supply in this appendix the proof that the functionals 
$\omega_{T_\sfa}$ on the algebra $\Ac(\Wc)$, generated by the 
free scalar massless field, which are defined in 
relations \eqref{kms1} and \eqref{kms2}, are KMS states 
with regard to the automorphic action of the 
dynamics $\alpha_{\Lambda_\sfa(t)}$, $t \in \RR$.
In doing so we rely on the well-known fact that
the state $\omega_{\sfa/2 \pi}$ coincides with the restriction
of the inertial vacuum state $\omega_0 \upharpoonright \Ac(\Wc)$, 
which satisfies the KMS-condition corresponding to the 
equilibrium parameter ${\sfa/2 \pi}$, c.f.  
\cite{Fulling, Sewell, BiWi}. 

\medskip
We begin by noting that the commutator function of the 
free massless theory can be presented in the form \  
$\kappa(f,g) = \langle f , g \rangle_0  - \langle g , f \rangle_0$,
where $ \langle f , g \rangle_0$
denotes the scalar product of the single particle vectors 
$|f\rangle_0, |g\rangle_0 \in \Hc_0$, \ $f,g \in \Dbox(\RR^4)$; 
these are generated  from~$\Omega_0$ by the smeared free field operators 
in the GNS representation \ $(\pi_0, \Hc_0, \Omega_0)$ 
induced by $\omega_0$. 
Restricting these quantities to test functions 
$f,g \in \Dbox(\Wc)$ we make use of the following 
facts.

\medskip 
\noindent (i) For any given $f,g \in \Dbox(\Wc)$, having relatively 
compact supports in $\Wc$, the support of $g_{\Lambda_\sfa(t)^{-1}}$
is timelike separated from the support of $f$ for sufficiently 
large $|t|$, cf.\ the proof of Lemma \ref{geometry}. Thus, because 
of Huygens' principle, the function 
$t \mapsto \kappa(f,g_{\Lambda_\sfa(t)^{-1}})$, $t \in \RR$, has compact support.
It is also arbitrarily often differentiable (since the free field 
is an operator-valued distribution) and hence is a test function. 

\medskip 
\noindent (ii) The boosts are 
unitarily implemented in the vacuum representation,
\viz $U_0(\Lambda_\sfa(t)) = e^{i t \sfa  K}$, \mbox{$t \in \RR$},
where $K$ is the canonical generator of the boosts. Since 
$\Omega_0$ is invariant under the action of these unitary operators, 
one obtains  $ \langle f , g_{\Lambda_\sfa(t)^{-1}} \rangle_0 
=  \langle f , e^{i t \sfa  K} g \rangle_0$, $t \in \RR$.
Moreover, as  $\omega_0 \upharpoonright \Ac(\Wc)$
satisfies the KMS-condition with respect to the action of 
$\alpha_{\Lambda_\sfa (t)}$ for the
equilibrium parameter $\sfa/2 \pi$, it follows that the 
vectors $|f\rangle_0, |g\rangle_0$ are in the domain of 
$e^{- \pi K}$ and that the KMS boundary condition can be 
presented for the single particle vectors in the form 
$\langle e^{- \pi K} f , e^{- \pi K} g_{\Lambda_\sfa(t)^{-1}} \rangle_0
= \langle g_{\Lambda_\sfa(t)^{-1}}, f \rangle_0$, $t \in \RR$.
So one arrives at the equality 
$\big( \langle f, e^{i t \sfa  K} g \rangle_0 -
\langle e^{- \pi K} f ,  e^{i t \sfa  K} e^{- \pi K} g \rangle_0 \big)
= \kappa(f, g_{\Lambda_\sfa(t)^{-1}})$, $t \in \RR$, for
any $f,g \in \Dbox(\Wc)$. 

\medskip
\noindent (iii) One now makes use of the spectral decomposition
of the selfadjoint generator $K$, which has Lebesgue absolutely  
continuous spectrum 
on the single particle space. Denoting by $E(k)$, $k \in \RR$, its
spectral resolution, one obtains the equality 
$\int \! e^{it \sfa k} (1 - e^{-2\pi k}) \, d \langle f, E(k) g \rangle_0
= \kappa(f, g_{\Lambda_\sfa(t)^{-1}})$, $t \in \RR$. Since  
$t \mapsto \kappa(f, g_{\Lambda_\sfa(t)^{-1}})$ is a test function, its 
Fourier transform \
\mbox{$k \mapsto 
\sqrt{2\pi} \, (1 - e^{-2\pi k/ \sfa}) \,  d \langle f, E(k/\sfa) 
g \rangle_0 / dk  $} \
is a (Schwartz) test function as well and vanishes at the origin
for any $f,g \in \Dbox(\Wc)$.  
Incidentally, this establishes the property used in the 
last part of the proof of Proposition \ref{centrallimits}. 

\medskip
\noindent (iv) Now the KMS boundary condition given above
implies  \ $e^{-2\pi k} \, d \langle f, E(k) g \rangle_0
= d \langle g, E(-k) f \rangle_0$ on $\RR$ 
for any $f, g \in \Dbox(\Wc)$. Since 
the function $k \mapsto (1 - e^{-2\pi k}) \, (1 - e^{- \sfa k/T_\sfa})^{-1}$
is positive and continuous for any $T_\sfa > 0$, it
follows that the real bilinear forms, given in \eqref{kms2}, 
can be presented as 
$$
v_{T_\sfa}(f,f) = 
\! \int \! (1 - e^{-2\pi k}) \, (1 - e^{- \sfa k/T_\sfa})^{-1} \, 
d \langle f, E(k) f \rangle_0 \, , \quad f \in \Dbox(\Wc) \, ,
$$
and hence are well-defined. Moreover, 
they can canonically be extended in both entries 
to  sesquilinear forms on the complex 
linear space $(\Dbox(\Wc) + i \, \Dbox(\Wc))$.
Since the functions $k \mapsto \langle f, E(k) f \rangle_0$ are 
continuous and monotonously increasing 
for any $f \in (\Dbox(\Wc) + i \, \Dbox(\Wc))$, 
they determine positive measures on~$\RR$. Hence~$v_{T_\sfa}$ 
defines a positive scalar product on 
$(\Dbox(\Wc) + i \, \Dbox(\Wc))$ for any $T_\sfa > 0$. 
By the KMS boundary condition 
and a change of variables
one also obtains for any $f,g \in \Dbox(\Wc)$ 
\begin{equation*}
\begin{split}
v_{T_\sfa}(f,g) - v_{T_\sfa}(g,f) = &
\int \! (1 - e^{-2\pi k}) \, (1 - e^{- \sfa k/T_\sfa})^{-1} \, 
\big( d \langle f, E(k) g \rangle_0  - d \langle g, E(k) f \rangle_0 \big) \\
= & \int \! \big( (1 - e^{-2\pi k}) \, (1 - e^{- \sfa k/T_\sfa})^{-1} 
- (1 - e^{ 2\pi k}) \, (1 - e^{\sfa k/T_\sfa})^{-1} \,  e^{- 2\pi k}  \big)  \, 
d \langle f, E(k) g \rangle_0  \\ 
=  & \int \! (1 - e^{-2\pi k}) \, d \langle f, E(k) g \rangle_0
\ = \ \kappa(f,g) \, .
\end{split}
\end{equation*}
Hence $v_{T_\sfa}$ defines a two-point function of the free massless
scalar field for any $T_\sfa > 0$.

\medskip
\noindent (v) It remains to prove that the functions 
$t \mapsto v_{T_\sfa}(f, g_{\Lambda_\sfa(t)^{-1}}) $ comply with the
KMS-condition. Making use of step (iv), 
one obtains for $f,g \in \Dbox(\Wc)$ and $t \in \RR$
$$
v_{T_\sfa}(f, g_{\Lambda_\sfa (t)^{-1}}) = 
\! \int \! e^{it \sfa k} \,  (1 - e^{- ak/T_\sfa})^{-1} \, (1 - e^{-2\pi k}) \, 
d \langle f, E(k) g \rangle_0  \, .
$$
Bearing in mind that \ $k \mapsto (1 - e^{-2\pi k}) \, 
d \langle f, E(k) g \rangle_0 / dk $ is a test function that
vanishes at $k = 0$, it is then apparent that 
$t \mapsto v_{T_\sfa}(f, g_{\Lambda_\sfa(t)^{-1}})$ can be continued to a function
that is  continuous and bounded 
on the strip $\{ z : 0 \leq \mbox{Im} \, z \leq 1/T_\sfa  \}$
and analytic in its interior. Its boundary value at the upper rim of
this strip is given by 
\begin{equation*}
\begin{split}
  & \int  \! e^{it \sfa k} \, e^{-ak/T_\sfa} \, 
(1 - e^{- ak/T_\sfa})^{-1} \, (1 - e^{-2\pi k}) \, 
d \langle f, E(k) g \rangle_0 \\
= & \int \! e^{-it \sfa k} \, e^{ak/T_\sfa} \, 
(1 - e^{ak/T_\sfa})^{-1} \, (1 - e^{2\pi k}) \, 
d \langle f, E(-k) g \rangle_0 \\
= & \int \! e^{-it \sfa k} \, 
(1 - e^{- ak/T_\sfa})^{-1} \, (1 - e^{-2\pi k}) \,  e^{- 2\pi k} \,
d \langle g, E(k) f \rangle_0  \  = \ 
v_{T_\sfa}(g_{\Lambda_\sfa (t)^{-1}},f) \, ,
\end{split}
\end{equation*}
proving the KMS property of $v_{T_\sfa}$, $T_\sfa > 0$. 

\medskip
\noindent (vi)
After these preparations it follows 
by standard arguments, cf. for example \cite{MaVe},
that the functionals~$\omega_{T_\sfa}$, defined in equation~\eqref{kms1} in terms
of the two-point function $v_{T_\sfa}$, are states on the 
algebra $\Ac(\Wc)$ of the free massless scalar field. For, making 
use of the Weyl relations, one has for 
any $ f_i, f_k \in \Dbox(\Wc)$, $i,k = 1, \dots, N$,
$$ 
\omega_{T_\sfa} (W(f_i)^* W(f_k)) =
e^{\kappa(f_i,f_k)/2} \, e^{-v_{T_a}(f_k-f_i,f_k-f_i)/2 }
= e^{ - v_{T_a}(f_i,f_i)/2} \, e^{v_{T_a}(f_i,f_k)} \,  e^{ - v_{T_a}(f_k,f_k)/2} \, . 
$$
Since the Hadamard products \ $v_{T_a}(f_i,f_k)^n$ 
of the positive matrix $v_{T_a}(f_i,f_k)$, $i,k = 1, \dots , N$, 
are again positive matrices for any $n \in \NN$, 
it is apparent that the functional $\omega_{T_\sfa}$ satisfies the condition of 
positivity. 
Moreover, the exponential function is entire analytic,
hence the functions
$$
t \mapsto \omega_{T_\sfa}(W(f) \, \alpha_{\Lambda_{\sfa}(t)}(W(g))
= e^{ - v_{T_a}(f,f)/2}  e^{- v_{T_a}(f,g_{\Lambda_\sfa (t)^{-1}} )} 
e^{ - v_{T_a}(g,g)/2} \, , \quad f,g \in \Dbox(\Wc) \, ,
$$
satisfy the KMS-condition for the dynamics $\alpha_{\Lambda_\sfa(t)}$,
$t \in \RR$, 
since the two-point function $v_{T_a}$ does. This completes the proof
that $ \omega_{T_\sfa}$ is a KMS state for any temperature
$T_\sfa > 0$.

\medskip
We conclude this appendix by noting that, by similar arguments, 
one can also establish the existence of a ground state 
for this dynamics in case of the free massless  scalar field.

\section*{Acknowledgment}
We acknowledge correspondence with Stephen A.\ Fulling,  
William G.\ Unruh and Robert M.\ Wald which induced us to 
include in this article remarks on the relation between 
our results and the Tolman--Ehrenfest~law.
We are also grateful to Ko Sanders for constructive 
comments on a preliminary version.

\end{document}